\documentclass[12pt]{article}
\usepackage{amssymb}
\usepackage{amsmath}
\usepackage{makeidx}
\usepackage{amsfonts}
\usepackage{mitpress}

\newtheorem{theorem}{Theorem}

\newtheorem{definition}[theorem]{Definition}
\newtheorem{example}[theorem]{Example}

\newtheorem{proposition}[theorem]{Proposition}

\newenvironment{proof}[1][Proof]{\noindent\textbf{#1.} }{\ \rule{0.5em}{0.5em}}

\newdimen\dummy
\dummy=\oddsidemargin
\title{A general quantum information model
 for the contextual dependent systems breaking the classical probability law}
\author{Masanari Asano\thanks{%
Department of Information Sciences, Tokyo University of Science, Yamasaki
2641, Noda-shi, Chiba, 278-8510 Japan}, Irina Basieva\thanks{%
Institute of Information Security, Russian State University for Humanities,
Moscow, Russia}, Andrei Khrennikov\thanks{%
International Center for Mathematical Modeling in Physics and Cognitive
Sciences, Linnaeus University, S-35195, V\"{a}xj\"{o}, Sweden}, Masanori Ohya%
$^{\ast }$ and Ichiro Yamato$^{\ast }$}

\begin{document}
\maketitle

\begin{abstract}
There exist several phenomena (systems) breaking the classical
probability laws. Such systems are contextual dependent adaptive
systems. In this paper, we
 present a new mathematical formula to compute the probability in those
systems by using the concepts of the adaptive dynamics and quantum information theory --
quantum channels and  the lifting. The basic examples of the contextual dependent phenomena
can be found in quantum physics. And recently similar examples were found in biological and
psychological sciences. Our novel approach is motivated by traditional quantum probability, but
it is general enough to describe aforementioned phenomena outside of quantum physics.
\end{abstract}


\section{Introduction}

There exist several phenomena (systems) breaking the classical
probability laws such as quantum interference, e.g., \cite{PL0}, \cite{PL},
(two slit experiment), quantum-like interference in cognitive
science, the game of prisoner's dilemma (PD game), the
lactose-glucose interference in E. coli growth. The quantum-like
statistical models in psychology and cognitive science has been
discussed in \cite{KHR7}--\cite{AKO1}. The PD game was considered
by taking account of the players' minds \cite{AKO1,AKO2}. The
lactose-glucose interference is studied as the quantum-like
interference \cite{BKOY}.

These phenomena (systems) will require us a change of classical
probability law, e.g., \cite{PL}. One of our trials is  to create
a new mathematical model which  will describe in the unified
framework both ``traditional quantum phenomena'' and recently
found quantum-like phenomena outside of physics, cf. \cite{PL2}. A new general
rule of our probabilistic model is the updating the Bayesian law
\cite{ABKO}. It is important to notice that these phenomena are
contextual dependent, so that they are adaptive to the context of
the surroundings.

In such systems, the conditional probability can not be defined in
usual mathematical framework. It is well known that in quantum
systems the conditional probability does not exist (see the
section 3) in the sense of classical systems, so that the naive
total probability law should be reconsidered. Same situation is
occurred even in non-quantum systems.

Let us consider a simple and intuitive example: When one takes
sugar S and chocolate C and he is asked whether it is sweet (1) or
not so (2). Then the simple classical probability law may not be
satisfied, that is,

\begin{equation*}
P(C=1)\neq P(C=1|S=1)P(S=1)+P(C=1|S=2)P(S=2)
\end{equation*}%
\quad because the LHS $P(C=1)$ will be very close to 1 but the RHS
will be less than $\frac{1}{2}$. After taking very sweet sugar, he
will taste the chocolate is not so sweet. Taking sugar changes his
taste, i.e., the situation of the tongue changes. The conditional
probability should be defined on the basis of such a change, so
that it is observable-adaptive quantity. The $P(C=\ast |S=\ast )$
should be written as $P$adap$(C=\ast |S=\ast )$ and its proper
mathematical description (definition) should be given, that is, we
will give a mathematical formula to compute the LHS and the RHS
above.

In this paper we apply the concept of the adaptive dynamics to
make a mathematical framework for the study of these contextual
dependent systems. We present adaptive dynamics in the framework
of quantum information theory by operating with quantum channels
and liftings of input/output states.

We also remark that application of quantum information theory
outside of quantum physics, e.g., for macroscopic biological
systems, wakes up again the long debate on a possibility to
combine the realistic and quantum descriptions, cf. \cite{GAR},
\cite{GAR1}, \cite{Theo}. At the moment we are not able to present
a consistent interpretation for coming applications of quantum
information theory outside of quantum physics; we can only keep
close to the {\it operational interpretation of quantum
information theory}, e.g., \cite{DEM}, \cite{DAR}.  In applications 
of quantum probability outside of physics, the {\it Bayesian approach}
to quantum  probability and information interpretation of the quantum 
state \cite{Fuchs},  \cite{Fuchs1} are the most natural.

\section{Adaptive Dynamics}

The idea of the adaptive dynamics has implicitly appeared in
series of papers \cite{O1,O2,AO2,ISO,OV1,OV2,IOV1,KOT,AO1,O4} for
the study of compound dynamics, chaos ansd the SAT algorithm. The
name of the adaptive dynamics was deliberately used in
\cite{O4}.The AD has two aspects, one of which is the
"observable-adaptive" and another is the "state-adaptive".

\emph{The observable-adaptive dynamics is a dynamics characterized as
follows: (1) Measurement depends on how to see an observable to be measured.
(2)The interaction between two systems depends on how a fixed observable
exists, that is, the interaction is related to some aspects of obsevables to
be measured or prepared.}

\emph{The state-adaptive dynamics is a dynamics characterized as
follows: (1)Measurement depends on how the state to be used
exists, as same as the observable. (2)The correlation between two
systems interaction depends on how the state of at least one of
the systems at one instant exists e.g., the interaction
Hamiltonian depends on the state at that.}

The idea of observable-adaptivity comes from studying chaos. We claimed that
any observation will be unrelated or even contradicted to mathematical
universalities such as taking limits, sup, inf, etc. Observation of chaos is
a result due to taking suitable scales of, for example, time, distance or
domain, and it will not be possible in the limiting cases. Examples of the
observable-adaptivity are used to understand chaos \cite{O2,KOT} and examine
the violation of Bell's inequality, namely the chameleon dynamics of Accardi
\cite{AIR}. The idea of the state-adaptivity is implicitly started in
constructing a compound state for quantum communication \cite{O,O1,O3,AO2} \
Examples of the state-adaptivity are seen in an algorithm solving NP
complete problem, i.e., a pending problem for more than 30 years asking
whether there exists an algorithm solving a NP complete problem in
polynomial time, as discussed \cite{OV1,OV2,AO1}.

We will discuss in the section 5 how we can apply the adaptive dynamics to a
bio-system or a psycho-system. The concept of the adaptivity is naturally
existed in such systems. Our formulation here contains some treatments shown
in the book \cite{OV2} to understand the evolution of HIV-1, the brain
function and the irrational behavior of prisoners.

\subsection{Conditional probability and joint probability in quantum systems}

The conditional probability and the joint probability do not generally exist
in quantum system, which is an essential difference from classical system.
First of all, let us fix the notations to be used throughout in this paper.
We will review these facts for the sequel uses.

Let $\mathcal{H}$, $\mathcal{K}$ be the Hilbert spaces describing the system
of interest, $\mathcal{S(H)}$ be the set of all states or probability
measures on $\mathcal{H}$, $\mathcal{O(H)}$ be the set of all observables or
events on $\mathcal{H}$ and $\mathcal{P}(\mathcal{H})\subset \mathcal{O(H)}$
be the set of projections in $\mathcal{O(H)}$.

In classical probability, the joint probability for two events $A$ and $B$
is
\begin{equation*}
\mu (A\cap B)
\end{equation*}%
and the conditional probability is defined by
\begin{equation*}
{\frac{{\mu (A\cap B)}}{{\mu (B)}}}.
\end{equation*}%
In quantum probability, if the von Neumann-L\"uder projection rule
is correct, after a measurement of $F\in
\mathcal{P}(\mathcal{H})$, a state $\rho $ is considered to be
\begin{equation*}
\rho _{F}={\frac{{F\rho F}}{\text{\textrm{tr}}{\rho F}}}.
\end{equation*}%
When we observe an event $E\in \mathcal{P}(\mathcal{H})$, the expectation
value becomes%
\begin{equation}
\text{\textrm{tr}}\rho _{F}E={\frac{\text{\textrm{tr}}{F\rho FE}}{\text{%
\textrm{tr}}{\rho F}}}={\frac{\text{\textrm{tr}}{\rho FEF}}{\text{\textrm{tr}%
}{\rho F}}}.  \label{Cond1}
\end{equation}%
This expectation value can be a candidate of the \textit{conditional
probability in QP (quantum probability).}

There is another candidate for the conditional probability in QP, which is a
direct generalization of CP (classical probability).

This alternative expression of joint probability and the conditional
probability in QP are expressed as
\begin{equation*}
\varphi (E\wedge F)
\end{equation*}%
and
\begin{equation}
{\frac{\varphi (E\wedge F)}{\varphi (F)}},  \label{Cond2}
\end{equation}%
where $\varphi $ is a state (a measure) and $\wedge $ is the meet
of two events (projections) corresponding to $\cap $ in CP, and
for the state describing by a density operator, we have
\begin{equation*}
\varphi (\cdot )=\text{\textrm{tr}}\rho (\cdot ).
\end{equation*}

We ask when the above two expressions (\ref{Cond1}) and (\ref{Cond2})\ in QP
are equivalent.\ From the next proposition, $\varphi (\cdot \wedge
F)/\varphi (F)$ is not a probability measure (state) on $\mathcal{P}(%
\mathcal{H})$.

\begin{proposition}
(1) When $E$ commutes with $F$, the above two expressions are equivalent,
namely,
\begin{equation*}
{\frac{{\varphi (FEF)}}{{\varphi (F)}}}={\frac{{\varphi (E\wedge F)}}{{%
\varphi (F)}}}.
\end{equation*}%
(2) When $EF\neq FE$, ${\frac{{\varphi (\cdot \wedge F)}}{{\varphi (F)}}}$
is not a probability on $\mathcal{P}_{\mathcal{H}}$, so that the above two
expressions are not equivalent.
\end{proposition}

\begin{proof}
(1) $EF=FE$ implies $E\wedge F=EF$ and $FEF=EFF=EF^{2}=EF$, so that
\begin{equation*}
{\frac{{\varphi \left( {E\wedge F}\right) }}{{\varphi \left( F\right) }}}={%
\frac{{\varphi \left( {FEF}\right) }}{{\varphi \left( F\right) }}}={\frac{{%
\varphi \left( {EF}\right) }}{{\varphi \left( F\right) }}}.
\end{equation*}

(2) Put $K_{\varphi }\left( {E\mid F}\right) \equiv {\frac{{\varphi \left( {%
E\wedge F}\right) }}{{\varphi \left( F\right) }}\ }$and put $z\in
linsp\left\{ x{,y}\right\} ,$ $z\neq x,y$ for any $x,y\in \mathcal{H}$. Take
the projections $P_{x}=\left\vert x\right\rangle \left\langle x\right\vert
,\;P_{y}=\left\vert y\right\rangle \left\langle y\right\vert
,\;P_{z}=\left\vert z\right\rangle \left\langle z\right\vert $ such that $%
\left( {P_{x}\vee P_{y}}\right) \wedge P_{z}=P_{z}$ and $P_{x}\wedge
P_{z}=0=P_{y}\wedge P_{z}$. Then%
\begin{equation*}
K_{\varphi }\left( \left( P_{x}\vee P_{y}\right) \wedge P_{z}\mid F\right)
=K_{\varphi }\left( P_{z}\mid F\right) \neq 0,
\end{equation*}

\begin{equation*}
K_{\varphi }{\left( {P_{x}\wedge P_{z}\mid F}\right) +}K_{\varphi }\left( {{%
P_{y}\wedge P_{z}\mid F}}\right) =0.
\end{equation*}
Therefore
\begin{equation*}
K_{\varphi }\left( {\left( {P_{x}\vee P_{y}}\right) \wedge P_{z}\mid F}%
\right) \neq K_{\varphi }\left( {P_{x}\wedge P_{z}\mid F}\right) +K_{\varphi
}\left( {P_{y}\wedge P_{z}\mid F}\right)
\end{equation*}
so that $K_{\varphi }\left( {\cdot \mid F}\right) $ is not a probability
measure on $\mathcal{P_{H}}$.
\end{proof}

In CP, the joint distribution for two random variables $f$ and $g$ is
expressed as
\begin{equation*}
\mu _{f,g}\left( {\Delta _{1},\Delta _{2}}\right) =\mu \left( {f^{-1}\left( {%
\Delta _{1}}\right) \cap g^{-1}\left( {\Delta _{2}}\right) }\right)
\end{equation*}%
for any Borel sets ${\Delta _{1},\Delta _{2}\in B\left( \mathbb{R}\right) }$%
. The corresponding quantum expression is either
\begin{equation*}
\varphi _{A,B}\left( {\Delta _{1},\Delta _{2}}\right) =\varphi \left( {%
E_{A}\left( {\Delta _{1}}\right) \wedge E_{B}\left( {\Delta _{2}}\right) }%
\right)
\end{equation*}%
or%
\begin{equation*}
\varphi ({E_{A}\left( {\Delta _{1}}\right) }\cdot {E_{B}\left( {\Delta _{2}}%
\right) })
\end{equation*}%
for two observables $A,\,B$ and their spectral measures $E_{A}(\cdot ),$ $%
E_{B}(\cdot )$ such that%
\begin{equation*}
A=\int aE_{A}\left( da\right) ,\text{ \ }B=\int bE_{B}\left( da\right) .
\end{equation*}%
It is easily checked that neither one of the above expressions satisfies
neither the condition of probability measure nor the marginal condition
unless $AB=BA$, so that they can not be the joint quantum probability in the
classical sense.

Let us explain the above situation, as an example, in a physical measurement
process. When an observable $A$ has a discrete decomposition like
\begin{equation*}
A=\sum_{k}{a_{k}F_{k}},\;{F_{i}\bot F_{j}\;\left( {i\neq j}\right) },
\end{equation*}%
the probability obtaining ${a_{k}}$ by measurement in a state $\rho $ is
\begin{equation*}
p_{k}=\text{\textrm{tr}}\rho F_{k}
\end{equation*}%
and the state $\rho $ is changed to a (conditional) state $\rho _{k}$ such
that
\begin{equation*}
\rho _{k}={\frac{{F_{k}\rho F_{k}}}{\text{\textrm{tr}}{\rho F_{k}}}}\equiv
P_{\rho }\left( {\cdot |{F_{k}}}\right) .
\end{equation*}%
After the measurement of $A$, we will measure a similar type observable $B$
(i.e., $B=\sum_{j}{b_{j}E_{j}},\;\left( {E_{i}\bot E_{j}\;}({i\neq j}\right)
$) and the probability obtaining ${b_{j}}$ after we have obtained the above $%
a_{k}$ for the measurement of $A$ is given by%
\begin{align}
{p_{jk}}& ={\left( \text{\textrm{tr}}{\rho F_{k}}\right) \left( \text{%
\textrm{tr}}{\rho _{k}E_{j}}\right) }  \notag \\
& =\text{\textrm{tr}}\rho F_{k}E_{j}F_{k}  \label{5-37} \\
& =P_{\rho }\left( E_{j}|F_{k}\right) \text{\textrm{tr}}\rho F_{k.}  \notag
\end{align}%
This $p_{jk}$ satisfies%
\begin{align}
\sum_{j,k}p_{jk}& =1,  \notag \\
\sum_{j}p_{jk}& =\text{\textrm{tr}}\rho F_{k}=p_{k},  \label{5-38}
\end{align}%
but not
\begin{equation*}
\sum_{k}{p_{jk}=\mathrm{tr}\rho E_{j}}
\end{equation*}%
unless $E_{j}F_{k}=F_{k}E_{j}\;(\forall j,k)$ so that ${p_{jk}}$ is not
considered as a joint quantum\textbf{\ }probability distribution. More
intutive expression breaking the usual classical probability law is the
following:

\begin{eqnarray*}
p_{jk} &=&P(B=b_{j}\mid A=a_{k})P(A=a_{k})\text{ }\cdots \text{and } \\
P(B &=&b_{j})\neq \sum_{k}P(B=b_{j}\mid A=a_{k})P(A=a_{k})
\end{eqnarray*}

\textit{Therefore we conclude in quantum system that the above two
candidates can not satisfy the properties of both conditional and joint
probabilities in the sense of classical system.}

The above discussion shows that the order of the measurement of two
observables $A$ and $B$ is essential and it gives us a different expectation
value, hence the state change.

\section{Liftng and joint probability}

In order to partially solve the difficulty of the nonexistence of
joint quantum distribution, the notion of compound state \cite{O1}
satisfying the marginal condition is useful. In this section we
discuss a bit general notion named \textquotedblright
lifting\textquotedblright \cite{AO1} to discuss new scheme of
probability containing both classical and quantum.

\begin{definition}
Let $\mathcal{A}_{1},\mathcal{A}_{2}$ be C*-algebras and let $\mathcal{A}%
_{1}\otimes \mathcal{A}_{2}$ be a fixed C*-tensor product of $\mathcal{A}%
_{1} $ and $\mathcal{A}_{2}$. A \textit{lifting} from $\mathcal{A}_{1}$ to $%
\mathcal{A}_{1}\otimes \mathcal{A}_{2}$ is a weak\ $\ast $-continuous map
\begin{equation*}
\mathcal{E}^{\ast }:\mathcal{S}(\mathcal{A}_{1})\rightarrow \mathcal{S}(%
\mathcal{A}_{1}\otimes \mathcal{A}_{2})
\end{equation*}%
If $\mathcal{E}^{\ast }$ is affine and its dual is a completely positive
map, we call it a linear lifting; if it maps pure states into pure states,
we call it pure.
\end{definition}

The algebras $\mathcal{A}_{1},$ $\mathcal{A}_{2}$ can be considered as two
systems of interest, for instance, $\mathcal{A}_{1}$ is an objective system
for a study and $\mathcal{A}_{2}$ is the subjective system or the surrouding
of $\mathcal{A}_{1}.$

Note that to every lifting from $\mathcal{A}_{1}$ to $\mathcal{A}_{1}\otimes
\mathcal{A}_{2}$ we can associate two channels: one from $\mathcal{A}_{1}$
to $\mathcal{A}_{1}$, defined by
\begin{equation*}
\Lambda ^{\ast }\rho _{1}(A_{1})\equiv (\mathcal{E}^{\ast }\rho
_{1})(A_{1}\otimes 1)\quad ;\quad \forall A_{1}\in \mathcal{A}_{1}
\end{equation*}%
another from $\mathcal{A}_{1}$ to $\mathcal{A}_{2}$, defined by
\begin{equation*}
\Lambda ^{\ast }\rho _{1}(A_{2})\equiv (\mathcal{E}^{\ast }\rho
_{1})(1\otimes A_{2})\quad ;\quad \forall A_{2}\in \mathcal{A}_{2}
\end{equation*}%
In general, a state $\varphi \in \mathcal{S}(\mathcal{A}_{1}\otimes \mathcal{%
A}_{2})$ such that
\begin{equation*}
\varphi \mid _{\mathcal{A}_{1}\otimes 1}=\rho _{1}\quad ;\quad \varphi \mid
_{1\otimes \mathcal{A}_{2}}=\rho _{2}
\end{equation*}%
is called a compound state of the states $\rho _{1}\in \mathfrak{S}(\mathcal{%
A}_{1})$ and $\rho _{2}\in \mathfrak{S}(\mathcal{A}_{2})$. Remark here that
the above compound state is nothing but the joint probability in CP.

The following problem is important in several applications: Given a state $%
\rho _{1}\in \mathcal{S}(\mathcal{A}_{1})$ and a channel $\Lambda ^{\ast }:%
\mathcal{S}(\mathcal{A}_{1})\rightarrow \mathcal{S}(\mathcal{A}_{2})$, find
a standard lifting $\mathcal{E}^{\ast }:\mathcal{S}(\mathcal{A}%
_{1})\rightarrow \mathcal{S}(\mathcal{A}_{1}\otimes \mathcal{A}_{2})$ such
that $\mathcal{E}^{\ast }\rho _{1}$ is a compound state of $\rho _{1}$ and $%
\Lambda ^{\ast }\rho _{1}$. Several particular solutions of this problem
have been proposed by Ohya, Ceccini and Petz, however an explicit
description of all the possible solutions to this problem is still missing,
which might be related to find a new scheme of probability theory.

\textit{However it is not sure that one can resolve the difficulty
of quantum probability if one can solve this problem. The compound
state corresponds to the joint probability in classical systems,
but there is still ambiguity to define the conditional state in
quantum systems. As pointed out in Introduction, the usual
conditional probability meets an inadequacy to interpret a certain
phenomenon, in which it is important not to manage to set the
conditional state by mimicking the classical one but to make a
mathematical rule to set new treatment of probabilistic aspects of
such a phenomenon.}

\begin{definition}
A lifting from $\mathcal{A}_{1}$ to $\mathcal{A}_{1}\otimes \mathcal{A}_{2}$
is called nondemolition for a state $\rho _{1}\in \mathcal{S}(\mathcal{A}%
_{1})$ if $\rho _{1}$ is invariant for $\Lambda ^{\ast }$ i.e., if for all $%
a_{1}\in \mathcal{A}_{1}$
\begin{equation*}
(\mathcal{E}^{\ast }\rho _{1})(a_{1}\otimes 1)=\rho _{1}(a_{1})
\end{equation*}%
The idea of this definition being that the interaction with system $2$ does
not alter the state of system $1$.
\end{definition}

\begin{definition}
A transition expectation from $\mathcal{A}_{1}\otimes \mathcal{A}_{2}$ to $%
\mathcal{A}_{1}$ is a completely positive linear map $\mathcal{E}:\mathcal{A}%
_{1}\otimes \mathcal{A}_{2}\rightarrow \mathcal{A}_{1}$ satisfying
\begin{equation*}
\mathcal{E}(1_{\mathcal{A}_{1}}\otimes 1_{\mathcal{A}_{2}})=1_{\mathcal{A}%
_{1}}.
\end{equation*}
\end{definition}

Let an initial state (resp. input signal) is changed (resp. transmitted) to
the final state resp. output state) due to a dynamics $\Lambda ^{\ast }$
(resp. channel). Here $\mathcal{A}_{1}$ (resp. $\mathcal{A}_{2}$) is
interpreted as the algebra of observables of the input (resp. output) system
and $\mathcal{E}^{\ast }$ describes the interaction between the input and
the ouput. If $\rho _{1}\in \mathcal{S}(\mathcal{A}_{1})$ is the initial
state, then the state $\rho _{2}=\Lambda ^{\ast }\rho _{1}\in \mathcal{S}(%
\mathcal{A}_{2})$ is the output state.

In several important applications, the state $\rho _{1}$ of the system
before the interaction (preparation, input signal) is not known and one
would like to know this state knowing only $\Lambda ^{\ast }\rho _{1}\in
\mathcal{S}(\mathcal{A}_{2})$, i.e., the state of the apparatus after the
interaction (output signal). From a mathematical point of view this problem
is not well posed, since the map $\Lambda ^{\ast }$ is usually not
invertible. The best one can do in such cases is to acquire a control on the
description of those input states which have the same image under $\Lambda
^{\ast }$ and then choose among them according to some statistical criterion.

Let us show some important examples of liftings and channels below

\begin{example}
\textbf{: Isometric lifting.}

Let $V:\mathcal{H}_{1}\rightarrow \mathcal{H}_{1}\otimes \mathcal{H}_{2}$ be
an isometry
\begin{equation*}
V^{\ast }V=1_{\mathcal{H}_{1}}.
\end{equation*}%
Then the map
\begin{equation*}
\mathcal{E}:x\in \mathbf{B}(\mathcal{H}_{1})\otimes \mathbf{B}(\mathcal{H}%
_{2})\rightarrow V^{\ast }xV\in \mathbf{B}(\mathcal{H}_{1})
\end{equation*}%
is a transition expectation in the sense of Accardi, and the associated
lifting maps a density matrix $w_{1}$ in $\mathcal{H}_{1}$ into
\begin{equation*}
\mathcal{E}^{\ast }w_{1}=Vw_{1}V^{\ast }
\end{equation*}%
in $\mathcal{H}_{1}\otimes \mathcal{H}_{2}.$ Liftings of this type are
called isometric. Every isometric lifting is a pure lifting. In this case
the channel $\Lambda ^{\ast }:\mathcal{H}_{1}\rightarrow \mathcal{H}_{1}$ is
given by tr$_{H_{2}}\mathcal{E}^{\ast }.$
\end{example}

It is the particular isometric lifting characterized by the properties.
\begin{equation*}
\mathcal{H}_{1}=\mathcal{H}_{2}=:\Gamma (\mathbb{C})\text{ }(\text{Fock
space over }\mathbb{C}\mathbf{)=}L^{2}\left( \mathbb{R}\right)
\end{equation*}%
\begin{equation*}
V:\Gamma (\mathbb{C})\rightarrow \Gamma (\mathbb{C})\otimes \Gamma (\mathbb{C%
})
\end{equation*}%
is characterized by the expression
\begin{equation*}
V\left\vert \theta \right\rangle =\left\vert \alpha \theta \right\rangle
\otimes \left\vert \beta \theta \right\rangle
\end{equation*}%
where $\left\vert \theta \right\rangle $ is the normalized coherent vector
parametrized by $\theta \in \mathbb{C}$ and $\alpha ,\beta \in \mathbb{C}$
are such that
\begin{equation*}
|\alpha |^{2}+|\beta |^{2}=1
\end{equation*}%
Notice that this liftings maps coherent states into products of coherent
states. So it maps the simplex of the so called classical states (i.e., the
convex combinations of coherent vectors) into itself. Restricted to these
states it is of convex product type explained below, but it is not of convex
product type on the set of all states.Denoting, for $\theta \in \mathbb{C}%
,\;\omega _{\theta }$ the coherent state on $\mathbf{B}(\Gamma (\mathbb{C}%
)), $ namely,
\begin{equation*}
\omega _{\theta }(b)=\left\langle \theta ,b\theta \right\rangle \ ;\ b\in
\mathbf{B}(\Gamma (\mathbb{C}))
\end{equation*}%
then for any $b\in \mathbf{B}(\Gamma (\mathbb{C}))$
\begin{equation*}
(\mathcal{E}^{\ast }\omega _{\theta })(b\otimes 1)=\omega _{\alpha \theta
}(b),
\end{equation*}%
so that this lifting is not nondemolition. These equations mean that, by the
effect of the interaction, a coherent signal (beam) $\left\vert \theta
\right\rangle $ splits into 2 signals (beams) still coherent, but of lower
intensity, but the total intensity (energy) is preserved by the
transformation.

Finally we mention two important beam splitting which are used to
discuss quantum gates and quantum teleportation.

(1) Superposed beam splitting:
\begin{equation*}
V_{s}\left\vert \theta \right\rangle \equiv {\frac{1}{\sqrt{2}}}(\left\vert
\alpha \theta \right\rangle \otimes \left\vert \beta \theta \right\rangle
-i\left\vert \beta \theta \right\rangle \otimes \left\vert \alpha \theta
\right\rangle )
\end{equation*}

(2) Beam splitting with two inputs and two output: Let $\left\vert \theta
\right\rangle $ and $\left\vert \gamma \right\rangle $ be two input coherent
vectors. Then

\begin{equation*}
V_{d}\left( \left\vert \theta \right\rangle \otimes \left\vert \gamma
\right\rangle \right) \equiv \left\vert \alpha \theta +\beta \gamma
\right\rangle \otimes \left\vert -\bar{\beta}\theta +\bar{\alpha}\gamma
\right\rangle ,
\end{equation*}%
where $\left\vert \alpha \right\vert ^{2}+\left\vert \beta \right\vert
^{2}=1 $. These extend linearly to isometry, and their isometric liftings
are neither of convex product type nor nondemolition type.

\begin{example}
\textit{Quantum measurement:}\emph{\ \ }If a measuring apparatus is prepared
by an positive operator valued measure $\left\{ Q_{n}\right\} $ then the
state $\rho $ changes to a state $\Lambda ^{\ast }\rho $ after this
measurement, $\rho \rightarrow \Lambda ^{\ast }\rho =\sum_{n}Q_{n}\rho
Q_{n}. $
\end{example}

\begin{example}
\textit{Reduction (Open system dynamics):}\emph{\ }If a system $\Sigma _{1}$
interacts with an external system $\Sigma _{2}$ described by another Hilbert
space $\mathcal{K}$ and the initial states of $\Sigma _{1}$ and $\Sigma _{2}$
are $\rho _{1}$ and $\rho _{2}$, respectively, then the combined state $%
\theta _{t}$ of $\Sigma _{1}$ and $\Sigma _{2}$ at time $t$ after the
interaction between two systems is given by%
\begin{equation*}
\theta _{t}\equiv U_{t}(\rho _{1}\otimes \rho _{2})U_{t}^{\ast },
\end{equation*}%
where $U_{t}=\exp (-itH)$ with the total Hamiltonian $H$ of $\Sigma _{1}$
and $\Sigma _{2}$. A channel is obtained by taking the partial trace w.r.t. $%
\mathcal{K}$ such as%
\begin{equation*}
\rho _{1}\rightarrow \Lambda ^{\ast }\rho _{1}\equiv \mathrm{tr}_{\mathcal{K}%
}\theta _{t}.
\end{equation*}

\begin{example}
\textbf{: }%
\index{lifting!compound}\textbf{The compound lifting.}
\end{example}

Let $\Lambda ^{\ast }:\mathcal{S}(\mathcal{A}_{1})\rightarrow \mathcal{S}(%
\mathcal{A}_{2})$ be a channel. For any $\rho _{1}\in \mathcal{S}(\mathcal{A}%
_{1})$ in the closed convex hull of the extremal states, fix a decomposition
of $\rho _{1}$ as a convex combination of extremal states in $\mathcal{S}(%
\mathcal{A}_{1})$
\begin{equation*}
\rho _{1}=\int_{\mathcal{S}(\mathcal{A}_{1})}\omega _{1}d\mu
\end{equation*}%
where $\mu $ is a Borel measure on $\mathcal{S}(\mathcal{A}_{1})$ with
support in the extremal states, and define
\begin{equation*}
\mathcal{E}^{\ast }\rho _{1}\equiv \int_{\mathcal{S}(\mathcal{A}_{1})}\omega
_{1}\otimes \Lambda ^{\ast }\omega _{1}d\mu
\end{equation*}%
Then $\mathcal{E}^{\ast }:\mathcal{S}(\mathcal{A}_{1})\rightarrow \mathcal{S}%
(\mathcal{A}_{1}\otimes \mathcal{A}_{2})$ is a lifting, nonlinear even if $%
\Lambda ^{\ast }$ is linear, and it is a nondemolition type.
\end{example}

The most general lifting, mapping $\mathcal{S}(\mathcal{A}_{1})$
into the closed convex hull of the extremal product states on
$\mathcal{A}_{1}\otimes \mathcal{A}_{2}$ is essentially of this
type. This nonlinear nondemolition lifting was first discussed by
Ohya to define the compound state and the mutual entropy for
quantum information communication \cite{O1,O3}. The above is a bit
more general because we shall weaken the condition that $\mu $ is
concentrated on the extremal states used in \cite{O1}.

Therefore once a channel is given, by which a lifting of convex product type
can be constructed. For example, the von Neumann quantum measurement process
is written, in the terminology of lifting, as follows: Having measured a
compact observable $A=\sum_{n}a_{n}P_{n}$ (spectral decomposition with $%
\sum_{n}P_{n}=I$) in a state $\rho $, the state after this measurement will
be
\begin{equation*}
\Lambda ^{\ast }\rho =\sum_{n}P_{n}\rho P_{n}
\end{equation*}%
and a lifting $\mathcal{E}^{\ast }$, of convex product type, associated to
this channel $\Lambda ^{\ast }$ and to a fixed decomposition of $\rho $ as $%
\rho $ $=\sum_{n}\mu _{n}\rho _{n}$ ($\rho _{n}\in \mathcal{S}(\mathcal{A}%
_{1})$) is given by :
\begin{equation*}
\mathcal{E}^{\ast }\rho =\sum_{n}\mu _{n}\rho _{n}\otimes \Lambda ^{\ast
}\rho _{n}.
\end{equation*}

\begin{example}
Amplifier channel: To recover the loss in the course of a quantum
communication, we need to amplify the signal (photon). In quantum optics, a
linear amplifier is usually expressed by means of annihilation operators $a$
and $b$ on $\mathcal{H}$ and $\mathcal{K}$, respectively :
\begin{equation*}
c=%
\sqrt{G}a\otimes I+\sqrt{G-1}I\otimes b^{\ast }
\end{equation*}%
where $G(\geq 1)$ is a constant and $c$ satisfies CCR (i.e.,
$[c,\;c^{\ast }]=I$) on $\mathcal{H}\otimes \mathcal{K}$. This
expression however is not convenient to compute several
information quantities, like entropy. The lifting expression of
the amplifier is good for such uses and it is given as follows:
Let $c=\mu a\otimes I+\nu I\otimes b^{\ast }$ with $\left\vert \mu
\right\vert ^{2}-\left\vert \nu \right\vert ^{2}=1$ and
$\left\vert \gamma \right\rangle $ be the eigenvector of $c$ :
$c\left\vert \gamma \right\rangle =\gamma \left\vert \gamma
\right\rangle $. For two coherent vectors $\left\vert \theta
\right\rangle $ on $\mathcal{H}$ and $\left\vert \theta ^{\prime
}\right\rangle $ on $\mathcal{K}$ , $\left\vert \gamma
\right\rangle $ can be written by the squeezing expression :
$\left\vert \gamma \right\rangle =\left\vert {\theta \otimes
\theta ^{\prime }\;;\;\mu ,\nu }\right\rangle $ and the lifting is
defined by an isometry
\begin{equation*}
V_{\theta ^{\prime }}\left\vert \theta \right\rangle =\left\vert {\theta
\otimes \theta ^{\prime }\;;\;\mu ,\nu }\right\rangle
\end{equation*}%
such that
\begin{equation*}
\mathcal{E}^{\ast }\rho =V_{\theta ^{\prime }}\rho V_{\theta ^{\prime
}}^{\ast }\quad \rho \in \mathcal{S}\left( \mathcal{H}\right) .
\end{equation*}%
The channel of the amplifier is
\begin{equation*}
\Lambda ^{\ast }\rho =\text{\textrm{tr}}_{\mathcal{K}}\mathcal{E}^{\ast
}\rho .
\end{equation*}
\end{example}

Finally we note that a channel is determined by a lifting and coversely a
lifting is constructed by a channel.

\section{New views of probability both in classical and quantum systems}

In this section we will discuss how to use the concept of lifting to explain
phenomena breaking the usual probability law.

Let $\mathcal{A},\mathcal{B}$ be C*-algebras describing the systems for a
study, more specifically, let $\mathcal{A},\mathcal{B}$ be the sets of all
observales in Hilbert spaces $\mathcal{H}$, $\mathcal{K}$; $\mathcal{A}=%
\mathcal{O(H)}$, $\mathcal{B}=\mathcal{O(K)}$. Let $\mathcal{E}^{\ast }$ be
a lifting from $\mathcal{S}\left( \mathcal{H}\right) $ to $\mathcal{S}\left(
\mathcal{H}\otimes \mathcal{K}\right) ,$so that its dual map $\mathcal{E}$
is a mapping from $\mathcal{A}\otimes \mathcal{B}$ to $\mathcal{A}$. There
are several liftings for various different cases to be considered: (1) If $%
\mathcal{K}$ is $\mathbb{C}$, then the lifting $\mathcal{E}^{\ast }$ is
nothing but a channel from $\mathcal{S}\left( \mathcal{H}\right) $ to $%
\mathcal{S}\left( \mathcal{H}\right) .$ (2) If $\mathcal{H}$ is $\mathbb{C}$%
, then the lifting $\mathcal{E}^{\ast }$ is a channel from $\mathcal{S}%
\left( \mathcal{H}\right) $ to $\mathcal{S}\left( \mathcal{K}\right) .$
Further $\mathcal{K}$ or $\mathcal{H}$ can be decomposed as $\mathcal{%
K=\otimes }_{i}\mathcal{K}_{i}$ (resp. $\oplus _{i}\mathcal{K}_{i}),$ and so
for $\mathcal{H}$, so that $\mathcal{B}$ can be $\otimes _{i}\mathcal{B}_{i}$
(resp. $\oplus _{i}\mathcal{B}_{i}$) and so for $\mathcal{A}$.

\textit{The adaptive dynamics is considered that the dynamics of a
state or an observable after an instant (say the time
}$t_{0})$\textit{\ attached to a system of interest is affected by
the existence of some other observable
and state at that instant.} Let $\rho $ $\in \mathcal{S}\left( \mathcal{H}%
\right) $ and $A\in \mathcal{A}$ be a state and an observable before $t_{0}$%
, and let $\sigma \in \mathcal{S}\left( \mathcal{H}\otimes \mathcal{K}%
\right) $ and $Q\in \mathcal{A}\otimes \mathcal{B}$ be a state and an
observable to give an effect to the state $\rho $ and the observable $A.$In
many cases, the effect to the state is dual to that to the observable, so
that we will discuss the effect to the state only. This effect is described
by a lifting $\mathcal{E}_{\sigma Q}^{\ast },$so that the state $\rho $
becomes $\mathcal{E}_{\sigma Q}^{\ast }\rho $ first, then it will be tr$_{%
\mathcal{K}}\mathcal{E}_{\sigma Q}^{\ast }\rho \equiv \rho _{\sigma Q}$. The
adaptive dynamics is the whole process such as

\begin{equation*}
Adaptive\text{ }Dynamics:\text{ }\rho \Rightarrow \mathcal{E}_{\sigma
Q}^{\ast }\rho \Rightarrow \rho _{\sigma Q}=tr_{\mathcal{K}}\mathcal{E}%
_{\sigma Q}^{\ast }\rho
\end{equation*}%
That is, what we need is how to construct the lifting for each problem to be
studied, that is, we properly construct the lifting $\mathcal{E}_{\sigma
Q}^{\ast }$ by choosing $\sigma $ and $Q$ properly.

The state change discussed in Section 2 is a naive example of the adaptive
dynamics, in which the lifting is given as follows: $Q=A=\sum_{k}{a_{k}F_{k}}%
\in \mathcal{A}$, ${F_{i}\bot F_{j}\;\left( {i\neq j}\right) }$ and $%
\mathcal{E}_{\sigma Q}^{\ast }\equiv \left\{ \mathcal{E}_{F_{k}A}^{\ast
}\right\} $ (this case the state $\sigma $ is not needed) such that for any $%
B$ $\in \mathcal{A}$
\begin{equation*}
P(B\mid A=a_{k})=trB\mathcal{E}_{F_{k}A}^{\ast }\rho =trBF_{k}\rho
F_{k}\diagup trF_{k}\rho .
\end{equation*}%
In this case the lifting is a channel from $\mathcal{S}\left( \mathcal{H}%
\right) $ to $\mathcal{S}\left( \mathcal{H}\right) ,$ the case (1)
above. It is true that we do not know whether this "projection
rule" can describe almost all probabilistic phenomena in nature or
not.

Let us go back to the discussion using lifting $\mathcal{E}_{\sigma Q}^{\ast
}$ above. The expectation value of another observable $B\in \mathcal{A}$ or $%
\mathcal{A}\otimes \mathcal{B}$ in the adaptive state $\rho _{\sigma Q}$ is
\begin{equation*}
tr\rho _{\sigma Q}B=tr_{\mathcal{H}}tr_{\mathcal{K}}B\mathcal{E}_{\sigma
Q}^{\ast }\rho .
\end{equation*}

Now suppose that there are two quantum event systems $A=\left\{ a_{k}\in
\mathbb{R},F_{k}\in \mathcal{A}\right\} $ and $B=\left\{ b_{j}\in \mathbb{R}%
,E_{j}\in \mathcal{A}\right\} ,$ where we do not assume $F_{k}$, $E_{j}$ are
projections, but they satisfy the conditions $\sum_{k}F_{k}=I,$ $%
\sum_{j}E_{j}=I$ as POVM (positive operator valued measure)
corresponding to the partition of a probability space in classical
system. Then the "joint-like" probability obtaining $a_{k}$ and
$b_{j}$ might be given by the \textit{formula}

\begin{equation}
P(a_{k},b_{j})=trE_{j}\boxdot F_{k}\mathcal{E}_{\sigma Q}^{\ast }\rho ,
\label{Cond.3}
\end{equation}%
where $\boxdot $ is a certain operation (relation) between $A$ and $B$, more
generally one can take a certain operator function $f(E_{j},F_{k})$ instead
of $E_{j}\boxdot F_{k}.$ If $\sigma ,Q$ are independent from any $F_{k}$, $%
E_{j}$ and the operation $\boxdot $ is the usual tensor product $\otimes $
so that $A$ and $B$ can be considered in two independent systems or to be
commutative, then the above "joint-like" probability becomes the joint
probability. However if not such a case, e.g., $Q$ is related to $A$ and $B,$
the situation will be more subtle. Therefore the problem is how to set the
operation $\boxdot $ and how to construct the lifting $\mathcal{E}_{\sigma
Q}^{\ast }$ in order to describe the particular problems associated to
systems of interest. We in the sequel discuss this problem in the contextual
dependent systems like bio-systems and psyco-systems mentioned in
Introduction. That is, we discuss how to apply the formula \ref{Cond.3} to
the following three problems breaking the usual probability law: (1) State
change of tongue for sweetness, (2) Lactose-glucose interference in E. coli
growth, (3) Updating the Bayesian law.

\section{State change of tongue for sweetness}

The first problem is not so sophisticated but very simple and common one. As
considered in Introduction, when one takes sugar S and chocolate C and he is
asked whether it is sweet (1) or not so (2). Then the simple classical
probability law may not be satisfied, that is,
\begin{equation*}
P(C=1)\neq P(C=1|S=1)P(S=1)+P(C=1|S=2)P(S=2)
\end{equation*}%
\quad because the LHS $P(C=1)$ will be very close to 1 but the RHS will be
less than $\frac{1}{2}$.

Let $e_{1}$ and $e_{2}$ be the orthogonal vectors describing sweet
and non-sweet states, respectively. \textit{The initial state of
tongue is neutral} such as
\begin{equation*}
\rho \equiv \left\vert x_{0}\right\rangle \left\langle x_{0}\right\vert ,
\end{equation*}%
where $x_{0}=\frac{1}{\sqrt{2}}\left( e_{1}+e_{2}\right) .$ Here we start
from the neutral pure state $\rho $ because we consider two sweet things$.$
It is enough for us to take the Hilbert space $\mathbb{C}^{2}$ for this
problem, so that $e_{1}$ and $e_{2}$ can be set as $\binom{1}{0}$ and $%
\binom{0}{1},$respectively.

When one takes "sugar", the operator corresponding to taking "sugar" will be
given as
\begin{equation*}
S=\left(
\begin{array}{cc}
\lambda _{1} & 0 \\
0 & \lambda _{2}%
\end{array}%
\right) ,
\end{equation*}%
where $\left\vert \lambda _{1}\right\vert ^{2}+$ $\left\vert \lambda
_{2}\right\vert ^{2}=1.$This operator can be regarded as the square root of
the sugar state $\sigma _{S};$
\begin{equation*}
\sigma _{S}=\left\vert \lambda _{1}\right\vert ^{2}E_{1}+\left\vert \lambda
_{2}\right\vert ^{2}E_{2},\text{ }E_{1}=\binom{1}{0}(10),E_{2}=\binom{0}{1}%
(01).
\end{equation*}%
Taking sugar, he will taste that it is sweet with the probability $%
\left\vert \lambda _{1}\right\vert ^{2}$ and non-sweet with the probability $%
\left\vert \lambda _{2}\right\vert ^{2},$ so $\left\vert \lambda
_{1}\right\vert ^{2}$ should be much higher than $\left\vert \lambda
_{2}\right\vert ^{2}$ for a usual sugar. This comes from the following
change of the neutral initial tongue (i.e., non-adaptive) state:
\begin{equation*}
\rho \rightarrow \rho _{S}=\Lambda _{S}^{\ast }(\rho )\equiv \frac{S^{\ast
}\rho S}{tr\left\vert S\right\vert ^{2}\rho },
\end{equation*}%
which is the state when he takes the sugar. This is similar to the
usual expression of state change in quantum dynamics, although it
is adaptive for sugar. The subtle point of the present problem is
that the state of tongue is neither $\rho _{S}$ nor $\rho $ at the
instant just after taking sugar. Note here that if we kill the
subjectivity (personal character?) of one's tongue, then the state
$\rho _{S}$ can be understood as
\begin{equation*}
E_{1}\rho _{S}E_{1}+E_{2}\rho _{S}E_{2},
\end{equation*}%
which is the unread objective state as usual in quantum measurement. We can
use the above two expressions $\rho _{S},$which give us the same result for
the computation of the probability.

For some time duration, the tongue becomes dull to sweetness, so
the tongue state can be written by means of a certain "exchanging"
operator $X=\left(
\begin{array}{cc}
0 & 1 \\
1 & 0%
\end{array}%
\right) $ such that%
\begin{equation*}
\rho _{S}^{a}=X\rho _{S}X,
\end{equation*}%
where $"a"$ means the adaptive change. Then similarly as sugar, when one
takes a chocolate, the state will be $\rho _{S\rightarrow C}^{a}$ given by
\begin{equation*}
\rho _{S\rightarrow C}^{a}=\Lambda _{C}^{\ast }(\rho _{S}^{a})\equiv \frac{%
C^{\ast }\rho _{S}^{a}C}{tr\left\vert C\right\vert ^{2}\rho _{S}^{a}},
\end{equation*}%
where $C$ will be given as
\begin{equation*}
C=\left(
\begin{array}{cc}
\mu _{1} & 0 \\
0 & \mu _{2}%
\end{array}%
\right)
\end{equation*}%
with $\left\vert \mu _{1}\right\vert ^{2}+$ $\left\vert \mu _{2}\right\vert
^{2}=1.$Common experience tells us that $\left\vert \lambda _{1}\right\vert
^{2}\geq $ $\left\vert \mu _{1}\right\vert ^{2}\geq \left\vert \mu
_{2}\right\vert ^{2}\geq \left\vert \lambda _{2}\right\vert ^{2}$ and the
first two are much larger than the last two.

As shown above, the adaptive set $\left\{ \sigma ,Q\right\} $ is the set $%
\left\{ S\text{ }\left( =\sigma _{S}\right) ,X,C\right\} $, we introduce the
following \textit{nonlinear demolition} lifting:
\begin{equation*}
\mathcal{E}_{\sigma Q}^{\ast }(\rho )(=\mathcal{E}_{S\text{ }\left( =\sigma
_{S}\right) XC}^{\ast }(\rho ))\equiv \rho _{S}\otimes \rho _{S\rightarrow
C}^{a}=\Lambda _{S}^{\ast }(\rho )\otimes \Lambda _{C}^{\ast }(X\Lambda
_{S}^{\ast }(\rho )X),
\end{equation*}%
which implies the joint probabilities $P(S=j,C=k)$ $(j,k=1,2)$ as
\begin{equation*}
P(S=j,C=k)=trE_{j}\otimes E_{k}\mathcal{E}_{\sigma Q}^{\ast }(\rho ).
\end{equation*}%
The probability that one tastes sweetness of the chocolate after tasting
sugar is
\begin{equation*}
P(C=1,S=1)+P(C=1,S=2)=\frac{\left\vert \lambda _{2}\right\vert
^{2}\left\vert \mu _{1}\right\vert ^{2}}{\left\vert \lambda _{2}\right\vert
^{2}\left\vert \mu _{1}\right\vert ^{2}+\left\vert \lambda _{1}\right\vert
^{2}\left\vert \mu _{2}\right\vert ^{2}}.
\end{equation*}%
Note that this probability is much less than%
\begin{equation*}
P(C=1)=trE_{1}\Lambda _{C}^{\ast }(\rho )=\left\vert \mu _{1}\right\vert
^{2},
\end{equation*}%
which is the probability of sweetness tasted by the neutral tongue $\rho $.
In this sense, the usual probability law%
\begin{equation*}
P(C=1)=P(S=1,C=1)+P(S=2,C=1)
\end{equation*}%
is not satisfied.

\section{Activity of lactose operon in E. coli}

The lactose operon is a group of genes in E. coli (Escherichia coli), and it
is required for the metabolism of lactose. This operon produces $\beta $%
-galactosidase, which is an enzyme to digest lactose into glucose and
galactose. There was an experiment measuring the activity of $\beta $%
-galactosidase which E. coli produces in the presence of (I) only 0.4\%
lactose, (II) only 0.4\% glucose, or (III) mixture 0.4\% lactose + 0.1\%
glucose, see \cite{IKA}. The activity is represented in Miller's units
(enzyme activity measurement condition), and it reaches to $3000$ units by
full induction. In the cases of (I) and (II), the data of $2920$ units and $%
33$ units were obtained. These results make one to expect that the
activity in the case (III) will be large, because the number of
molecules of lactose is larger than that of glucose. However, the
obtained data were only $43$ units. This result implies that E.
coli metabolizes glucose in preference to lactose. In biology,
this functionality of E. coli have been discussed, and it was
known that glucose has a property reducing lactose permease
provided by the operon. Apart from such qualitative and
biochemical explanation, it will be also necessary to discuss a
mathematical interpretation such that
the biological activity in E. coli is evaluated quantitatively. In the paper~%
\cite{BKOY}, it is pointed out that the activity of E. coli
violates the total probability law as shown below, which might
come from the preference in E. coli's metabolism. We will explain
this contextual behaivor by the formula \ref{Cond.3}.

We consider two events $L$ and $E$; $L$: E. coli detects a lactose molecule
in his metabolism and $G$: E. coli detects a glucose molecule. In the case
of (I) or (II), the probability $P(L)=1$ or $P(G)=1$ is given. In the case
of (III), $P(L)$ and $P(G)$ are calculated as
\begin{equation*}
P(L)=\frac{0.4}{0.4+0.1}=0.8,\ P(G)=\frac{0.1}{0.4+0.1}=0.2.
\end{equation*}%
The events $L$ and $G$ are exclusive each other, so it is assumed that $%
P(L\cup G)=P(L)+P(G)=1$. Further, we consider the events $\{+,-\}$ which
means that E. coli activates his lactose operon or not. From the
experimental data of the cases (I) and (II), the following conditional
probabilities are obtained:
\begin{eqnarray}
&&\mbox{(I): }P(+|L)=\frac{2920}{3000},  \notag \\
&&\mbox{(II): }P(+|G)=\frac{33}{3000}.  \label{CASEI}
\end{eqnarray}
In the case (III), if the total probability law is satisfied, the
probability $P(+\cap (L\cup G))$ is computed as
\begin{eqnarray*}
P(+\cap (L\cup G)) &=&P(+|L\cup G)P(L\cup G)=P(+|L\cup G) \\
&=&P(+\cap L)+P(+\cap G)=P(+|L)P(L)+P(+|G)P(G)
\end{eqnarray*}%
However, from the experimental data, one can estimate
\begin{equation*}
P(+|L\cup G)=\frac{43}{3000},
\end{equation*}%
so that the total probability law is violated:
\begin{equation*}
P(+|L\cup G)\neq P(+|L)P(L)+P(+|G)P(G).
\end{equation*}

\textit{We will explain how to compute the probabilities above by using the
concept of lifting}. Firstly, we introduce the initial state $\rho
=\left\vert x_{0}\right\rangle \left\langle x_{0}\right\vert $ on Hilbert
space $\mathcal{H}=\mathbb{C}^{2}$. The state vector $x_{0}$ is written by
\begin{equation*}
x_{0}=\frac{1}{\sqrt{2}}e_{1}+\frac{1}{\sqrt{2}}e_{2}.
\end{equation*}%
The basis $\{e_{1},e_{2}\}$ denote the detection of lactose or
glucose by E.
coli, that is, the events, $L$ and $G$. The E. coli at the initial state $%
\rho $ has not recognized the existence of lactose and glucose
yet. When the E. coli recognizes them, the following state change
occurs;
\begin{equation*}
\rho \mapsto \rho _{D}=\Lambda _{D}^{\ast }(\rho )\equiv \frac{D\rho D^{\ast
}}{\mathrm{tr}(|D|^{2}\rho )},
\end{equation*}%
where
\begin{equation*}
D=\left(
\begin{array}{cc}
\alpha & 0 \\
0 & \beta%
\end{array}%
\right)
\end{equation*}%
with $|\alpha |^{2}+|\beta |^{2}=1$. Note that $|\alpha |^{2}$ and
$|\beta |^{2}$ implies the probabilities for the events $L$ and
$G$, that is, $P(L)$ and $P(G)$. The state $\sigma _{D}\equiv
DD^{\ast }$ means the distribution of $P(L)$ and $P(G)$. In this
sense, the state $\sigma _{D}$ represents the solution of lactose
and glucose. We call $D$ the \textit{detection operator} and call
$\rho _{D}$ the \textit{detection state}. The state determining
the
activation of the operon in E.coli depends on the detection state $\rho _{D}$%
. We give such state by
\begin{equation*}
\rho _{op}=\Lambda _{Q}^{\ast }(\rho _{D})\equiv \frac{Q\rho _{D}Q^{\ast }}{%
\mathrm{tr}(Q\rho Q^{\ast })},
\end{equation*}%
where the operator $Q$ is written as
\begin{equation*}
Q=\left(
\begin{array}{cc}
a & b \\
c & d%
\end{array}%
\right) .
\end{equation*}%
We call $\rho _{op}$ the \textit{activation state} for the operon and call $%
Q $ the \textit{activation operator}. (The components $a,b,c$ and $d$ are
discussed later.)

We introduce the lifting
\begin{equation*}
\mathcal{E}_{D,Q}^{\ast }(\rho )=\Lambda _{Q}^{\ast }(\Lambda _{D}^{\ast
}(\rho ))\otimes \Lambda _{D}^{\ast }(\rho )\in \mathcal{K}\otimes \mathcal{H%
}=\mathbb{C}^{2}\otimes \mathbb{C}^{2},
\end{equation*}%
by which we can describe the correlation between the activity of lactose
operon and the ratio of concentration of lactose and glucose. For example,
the joint probabilities $P(+\cap L)$ and $P(-\cap L)$ are given by
\begin{eqnarray}
&&P(+\cap L)=\mathrm{tr}(E_{1}\otimes E_{1}\mathcal{E}_{D,Q}^{\ast }(\rho )),
\notag \\
&&P(-\cap L)=\mathrm{tr}(E_{2}\otimes E_{1}\mathcal{E}_{D,Q}^{\ast }(\rho )).
\label{JoP}
\end{eqnarray}%
Here, let us consider the case of $P(L)=|\alpha |^{2}=1$, then the
probabilities $P(\pm |L\cup G)=P(\pm \cap (L\cup G))=P(\pm \cap L)+P(\pm
\cap G)$ correspond to $P(\pm |L)$ of Eq.~(\ref{CASEI})-(I), and the Eq.~(%
\ref{JoP}) gives the conditional probabilities as
\begin{equation*}
P(+|L)=\frac{|a|^{2}}{|a|^{2}+|c|^{2}},\ P(-|L)=\frac{|c|^{2}}{%
|a|^{2}+|c|^{2}}.
\end{equation*}%
From these results, we may give the following forms for $a$ and $b$.
\begin{equation*}
a=\sqrt{P(+|L)}\mathrm{e}^{\mathrm{i}\theta _{+L}}k_{L},\ c=\sqrt{P(-|L)}%
\mathrm{e}^{\mathrm{i}\theta _{-L}}k_{L}
\end{equation*}%
Here, $k_{L}$ is a certain real number. In a similar way, we obtain
\begin{equation*}
b=\sqrt{P(+|G)}\mathrm{e}^{\mathrm{i}\theta _{+G}}k_{G},\ d=\sqrt{P(-|G)}%
\mathrm{e}^{\mathrm{i}\theta _{-G}}k_{G}
\end{equation*}%
for the components $b$ and $d$. To simplify the discussion, hereafter, we
assume $\theta _{+L}=\theta _{-L}$, $\theta _{+G}=\theta _{-G}$ and denote $%
\mathrm{e}^{\mathrm{i}\theta _{L}}k_{L}$, $\mathrm{e}^{\mathrm{i}\theta
_{G}}k_{G}$ by $\tilde{k}_{L}$, $\tilde{k}_{G}$. Then, the operator $Q$ is
decomposed to
\begin{equation}
Q=\left(
\begin{array}{cc}
\sqrt{P(+|L)} & \sqrt{P(+|G)} \\
\sqrt{P(-|L)} & \sqrt{P(-|G)}%
\end{array}%
\right) \left(
\begin{array}{cc}
\tilde{k}_{L} & 0 \\
0 & \tilde{k}_{G}%
\end{array}%
\right) .  \label{Qop}
\end{equation}%
The probability $P(+|L\cap G)$ with this $Q$ is calculated as
\begin{equation*}
P(+|L\cap G)=\frac{|\sqrt{P(+|L)}\tilde{k}_{L}\alpha +\sqrt{P(+|G)}\tilde{k}%
_{G}\beta |^{2}}{|\sqrt{P(+|L)}\tilde{k}_{L}\alpha +\sqrt{P(+|G)}\tilde{k}%
_{G}\beta |^{2}+|\sqrt{P(-|L)}\tilde{k}_{L}\alpha +\sqrt{P(-|G)}\tilde{k}%
_{G}\beta |^{2}}.
\end{equation*}%
This gives us how to compute the probability. The rate $|\tilde{k}_{L}|/|%
\tilde{k}_{G}|$ essentially determines the degree of the violation of the
total probability law. Recall the experimental data in the case (III). In
this case, $P(L)=|\alpha |^{2}=0.8<P(G)=|\beta |^{2}=0.2$, but $P(P(+|L\cap
G))$ is very small. According to our interpretation, it implies that the
rate $|\tilde{k}_{L}|/|\tilde{k}_{G}|$ is very small. In this sense, the
operator $F$ =$\left(
\begin{array}{cc}
\tilde{k}_{L} & 0 \\
0 & \tilde{k}_{G}%
\end{array}%
\right) $ in Eq.~(\ref{Qop}) specifies the preference in E. coli's
metabolism. We call $F$ the \textit{preference operator}. Finally note that
if $\alpha ,\beta $ are real and $\tilde{k}_{L}=\tilde{k}_{G}^{\ast }$, the
usual total probability law is held.

\section{ Bayesian Updating Biased By Psychological Factor}

The Bayesian updating is an important concept in Bayesian statics, and it is
used to describe a process of inference, which is explained as follows:
Consider two event systems denoted by $S_{1}=\{A,B\}$ and $S_{2}=\{C,D\}$,
where the events $A$ and $B$ are mutually exclusive, and the same holds for $%
C$ and $D$. Firstly, a decision-making entity, say Alice, estimates the
probabilities $P(A)$ and $P(B)$ for the events $A$ and $B$, which are called
the \textit{prior probabilities}. The prior probability is sometime called
\textquotedblleft \textit{subjective} probability" or \textquotedblleft
\textit{personal} probability". Further, Alice knows the conditional
probabilities $P(C|A)$ and $P(C|B)$ which are obtained from some statistical
data. When Alice sees the occurrence of the event $C$ or $D$ in the system $%
S_{2}$, she can change her prior prediction $P(A)$ and $P(B)$ to the
following conditional probabilities by \textit{Bayes' rule}: When Alice sees
the occurrence of $C$ in $S_{2}$, she can update her prediction for the
event $A$ from $P(A)$ to
\begin{equation*}
P(A|C)=\frac{P(C|A)P(A)}{P(C|A)P(A)+P(C|B)P(B)}.
\end{equation*}%
When Alice sees the occurrence of $D$ in $S_{2}$, she can update her
prediction to
\begin{equation*}
P(A|D)=\frac{P(D|A)P(A)}{P(D|A)P(A)+P(D|B)P(B)}.
\end{equation*}%
These conditional probabilities are called the \textit{posterior
probabilities}. The change of prediction is described as an
\textquotedblleft updating" from the prior probabilities $P(A)$ to the
posterior probability, and it is called the Bayesian updating.

In the paper~\cite{AOTKB}, we redescribed the process of Bayesian updating
in the framework of \textquotedblleft quantum-like representation", where we
introduced the following state vector defined on Hilbert space $\mathcal{H}=%
\mathcal{H}_{1}\otimes \mathcal{H}_{2}=\mathbb{C}^{2}\otimes \mathbb{C}^{2}$%
;
\begin{eqnarray}
\left\vert \Phi \right\rangle &=&\sqrt{P(A^{\prime })}\left\vert A^{\prime
}\right\rangle \otimes (\sqrt{P(C^{\prime }|A^{\prime })}\left\vert
C^{\prime }\right\rangle +\sqrt{P(D^{\prime }|A^{\prime })}\left\vert
D^{\prime }\right\rangle )  \notag \\
&&+\sqrt{P(B^{\prime })}\left\vert B^{\prime }\right\rangle \otimes (\sqrt{%
P(C^{\prime }|B^{\prime })}\left\vert C^{\prime }\right\rangle +\sqrt{%
P(D^{\prime }|B^{\prime })}\left\vert D^{\prime }\right\rangle ).
\label{Pstate2}
\end{eqnarray}%
We call this vector the \textit{prediction state vector}. The set of vectors
$\{\left\vert A^{\prime }\right\rangle ,\left\vert B^{\prime }\right\rangle
\}$ becomes an orthogonal basis on $\mathcal{H}_{1}$, and $\{\left\vert
C^{\prime }\right\rangle ,\left\vert D^{\prime }\right\rangle \}$ is another
orthogonal basis on $\mathcal{H}_{2}$. The $A^{\prime }$, $B^{\prime }$, $%
C^{\prime }$ and $D^{\prime }$ represent the events defined as

\begin{description}
\item[Event $A^{\prime }$:] Alice \textbf{judges} ``\textit{the event $A$
occurs in the system $S_1$}."

\item[Event $B^{\prime }$:] Alice \textbf{judges} ``\textit{the event $B$
occurs in the system $S_1$}."

\item[Event $C^{\prime }$:] Alice \textbf{judges} ``\textit{the event $C$
occurs in the system $S_2$}."

\item[Event $D^{\prime }$:] Alice \textbf{judges} \textquotedblleft \textit{%
the event $D$ occurs in the system $S_{2}$}."
\end{description}

These events are the \textit{subjective} events (judgments) in \textit{%
Alice's mentality} and the vectors $\left\vert A^{\prime }\right\rangle $, $%
\left\vert B^{\prime }\right\rangle $, $\left\vert C^{\prime }\right\rangle $
and $\left\vert D^{\prime }\right\rangle $ mean the \textit{mind}s toward
the above judgements. The vector $\left\vert \Phi \right\rangle $ represents
that these minds coexists in Alice's mentality. For example, Alice is
conscious of the mind $\left\vert A^{\prime }\right\rangle $ with the weight
$\sqrt{P(A^{\prime })}$, and under the condition of the event $A^{\prime }$,
she gives the weights $\sqrt{P(C^{\prime }|A^{\prime })}$ and $\sqrt{%
P(D^{\prime }|A^{\prime })}$ for the minds $\left\vert C^{\prime
}\right\rangle $ and $\left\vert D^{\prime }\right\rangle $. Such an
assignment of weights implies that Alice feels \textit{causality} between $%
S_{1}$ and $S_{2}$: The events in $S_{1}$ are \textit{causes} and the events
in $S_{2}$ are \textit{results}. The square of $\sqrt{P(A^{\prime })}$ is
equivalent to a prior probability $P(A)$ in the Bayesian theory. If Alice
knows the objective conditional probabilities $P(C|A)$ and $P(C|B)$, Alice
can give the weights of $\sqrt{P(C^{\prime }|A^{\prime })}$ and $\sqrt{%
P(C^{\prime }|B^{\prime })}$ as $P(C^{\prime }|A^{\prime })=P(C|A)$ and $%
P(C^{\prime }|B^{\prime })=P(C|B)$.

The process of the above Bayesian updating is represented in the term of the
lifting discussed in the previous section.\textit{\ }When Alice has \textit{%
the prediction state} $\left\vert \Phi \right\rangle \left\langle \Phi
\right\vert \equiv \rho $ and sees the occurrence of the event $C$ in $S_{2}$%
, her mind for the event $D^{\prime }$ is \textit{vanished} instantaneously.
This vanishing is represented as the reduction by the projection operator $%
M_{C^{\prime }}=I\otimes \left\vert C^{\prime }\right\rangle \left\langle
C^{\prime }\right\vert $;
\begin{equation*}
\frac{M_{C^{\prime }}\rho M_{C^{\prime }}}{\mathrm{tr}(M_{C^{\prime }}\rho )}%
\equiv \rho _{C^{\prime }}
\end{equation*}%
The posterior probability $P(A|C)$ is calculated by
\begin{equation*}
\mathrm{tr}(M_{A^{\prime }}\rho _{C^{\prime }}),
\end{equation*}%
where $M_{A^{\prime }}=\left\vert A^{\prime }\right\rangle \left\langle
A^{\prime }\right\vert \otimes I$.

The inference based on the Bayesian updating is rational from the
view point of probability theory. However, we empirically know
that our behavior is sometimes irrational in various situations
from the mathematical view point and we have psychological factors
that frequently disturb rational inferences. \textit{Such
irrational inference can be mathematically discussed by using the
concept of lifting as the examples above.} Let us
introduce the lifting for this problem from $S(\mathcal{H})$ to $S(\mathcal{H%
}\otimes \mathcal{K})$ by
\begin{equation*}
\mathcal{E}_{\sigma ,V}^{\ast }(\rho )=V\rho \otimes \sigma V^{\ast }.
\end{equation*}%
Here $\sigma \in S(\mathcal{K})$ is a state specifying a psychological
factor which Alice holds in her mentality. The operator $V$ on $\mathcal{H}%
\otimes \mathcal{K}$ is isometry and gives a correlation between the
prediction state $\rho $ and the psychological factor $\sigma $, in other
words, it specifies a psychological affection to Alice's rational inference.
We call the state defined by
\begin{equation*}
\tilde{\rho}\equiv \mathrm{tr}_{\mathcal{K}}(\mathcal{E}_{\sigma ,V}^{\ast
}(\rho )),
\end{equation*}%
\textit{the prediction state biased from the rational prediction }$\rho $.
The posterior probability is generally given by
\begin{equation*}
\tilde{P}(A,C)\equiv \mathrm{tr}\left( M_{A^{\prime }}M_{C^{\prime }}\tilde{%
\rho}M_{C^{\prime }}\right) ,
\end{equation*}%
so that
\begin{equation*}
\tilde{P}(A|C)=\mathrm{tr}\left( M_{A^{\prime }}\frac{M_{C^{\prime }}\tilde{%
\rho}M_{C^{\prime }}}{\mathrm{tr}(M_{C^{\prime }}\tilde{\rho})}\right) ,
\end{equation*}%
which is not equal to the original $P(A|C)$ in the rational inference.

Here, it should be noted that the prior probability from $\tilde{\rho}$,
which is given by $\tilde{P}(A)=\mathrm{tr}_{\mathcal{H}}(M_{A^{\prime }}%
\tilde{\rho})$, is not equal to the original $P(A)=\mathrm{tr}_{\mathcal{H}%
}(M_{A^{\prime }}{\rho })$ in general. However, the prior probability is
given and fixed before the updating, and if the psychological factor affects
only the estimation of posterior probabilities, then the prediction state $%
\tilde{\rho}$ can be set to satisfy the condition
\begin{equation*}
\tilde{P}(A)=P(A).
\end{equation*}%
In the paper~\cite{PRE}, we gave an example of such a state $\tilde{\rho}$
which is provided by the psychological factor called the \textit{reliability
of information}. This example has the following equalities
\begin{eqnarray*}
&&{P}(A)=\tilde{P}(A,C)+\tilde{P}(A,D)=\tilde{P}(A), \\
&&{P}(B)=\tilde{P}(B,C)+\tilde{P}(B,D)=\tilde{P}(B),
\end{eqnarray*}%
for the joint probabilities given by
\begin{eqnarray*}
\tilde{P}(A,C) &=&\mathrm{tr}_{\mathcal{H}}\left( M_{A^{\prime
}}M_{C^{\prime }}\tilde{\rho}\right) ,\ \tilde{P}(B,C)=\mathrm{tr}_{\mathcal{%
H}}\left( M_{B^{\prime }}M_{C^{\prime }}\tilde{\rho}\right) , \\
\tilde{P}(A,D) &=&\mathrm{tr}_{\mathcal{H}}\left( M_{A^{\prime
}}M_{D^{\prime }}\tilde{\rho}\right) ,\ \tilde{P}(B,D)=\mathrm{tr}_{\mathcal{%
H}}\left( M_{B^{\prime }}M_{C^{\prime }}\tilde{\rho}\right) .
\end{eqnarray*}%
On the other hand, for $P(C)=\mathrm{tr}_{\mathcal{H}}(M_{C^{\prime }}\rho )$
and $P(D)=\mathrm{tr}_{\mathcal{H}}(M_{D^{\prime }}\rho )$,
\begin{eqnarray}
&&{P}(C)\neq \tilde{P}(A,C)+\tilde{P}(B,C)=\tilde{P}(C|A)P(A)+\tilde{P}%
(C|B)P(B),  \notag \\
&&{P}(D)\neq \tilde{P}(A,D)+\tilde{P}(B,D)=\tilde{P}(D|A)P(A)+\tilde{P}%
(D|B)P(B),  \label{viop}
\end{eqnarray}%
that is, $P(C)\neq \tilde{P}(C)$ and $P(D)\neq \tilde{P}(D)$. Note that $%
P(C) $ means the probability of the event $C$ which will be estimated by
Alice with $\rho $. Alice with the biased $\tilde{\rho}$ will estimate $%
\tilde{P}(C)$, and then, the violation of total probability law of Eq.~(\ref%
{viop}) is not occurred actually. However, Eq.~(\ref{viop}) represents the
violation of Alice's rationality, and the difference between $P(C)$ and $%
\tilde{P}(C)$ specifies the degree of Alice's irrationality.


\begin{thebibliography}{99}

\bibitem{PL0} Plotnitsky, A.: Reading Bohr: Physics and Philosophy.  Springer,
Heidelberg-Berlin-New York (2006).

\bibitem{PL} Plotnitsky, A.: Epistemology and Probability: Bohr, Heisenberg, Schr�dinger,
and the Nature of Quantum-Theoretical Thinking.  Springer,
Heidelberg-Berlin-New York (2009).

\bibitem{KHR7} A. Khrennikov, \textit{Open Systems and Information Dynamics}
\textbf{11} (3), 267-275 (2004).

\bibitem{Khrennikov/QLBrain} A. Khrennikov, \emph{BioSystems} \textbf{84},
225--241 (2006).

\bibitem{Wolfgang} K.-H.Fichtner, L.Fichtner, W.Freudenberg and M.Ohya, On a
quantum model of the recognition process. \textit{QP-PQ:Quantum Prob. White
Noise Analysis} \textbf{21}, 64-84 (2008).

\bibitem{Jerome1} J. B. Busemeyer,   Z. Wang, and J. T. Townsend,  Quantum
dynamics of human decision making. J. Math. Psychology 50, 220-241
(2006)

\bibitem{Jerome2} J. R. Busemeyer,  M. Matthews, and Z. Wang : A Quantum
Information Processing Explanation of Disjunction Effects. In:
Sun, R. and Myake, N. (eds.) The 29th Annual Conference of the
Cognitive Science Society and the 5th International Conference of
Cognitive Science (Pp. 131-135) Mahwah, NJ. Erlbaum (2006)

\bibitem{Jerome3} J. R. Busemeyer,  E. Santuy, A. Lambert-Mogiliansky,
Comparison of Markov and quantum models of decision making. In P.
Bruza, W. Lawless, K. van Rijsbergen, D. A. Sofge, B. Coeke, S.
Clark (Eds.) Quantum interaction: Proceedings of the Second
Quantum Interaction Symposium, pp.68-74. London: College
Publications, (2008)

\bibitem{Cheon} T. Cheon and T. Takahashi,  Interference and inequality in quantum decision theory,
Phys. Lett. A  375 (2010) 100-104.

\bibitem{QG2} T. Cheon and I. Tsutsui, Classical and quantum contents of solvable game theory on Hilbert space,
Phys. Lett. A 348 (2006) 147-152.

\bibitem{ACARDI0} L. Accardi, A. Khrennikov, M. Ohya, The problem of
quantum-like representation in economy, cognitive science, and genetics.
In.: \textit{Quantum Bio-Informatics II: From Quantum Information to
Bio-Informatics.} L. Accardi, W. Freudenberg, M. Ohya, eds., p. 1-8, WSP,
Singapore (2008).

\bibitem{ACARDI} L. Accardi, A. Khrennikov, M. Ohya, Quantum Markov Model
for Data from. Shafir-Tversky Experiments in Cognitive Psychology. \textit{%
Open Systems and Information Dynamics,} \textbf{16}, 371E85 (2009).

\bibitem{Conte} Conte E., Khrennikov A., Todarello O., Federici A., Zbilut
J. P. Mental States Follow Quantum Mechanics during Perception and Cognition
of Ambiguous Figures. \textit{Open Systems and Information Dynamics,}
\textbf{16}, 1-17 (2009).

\bibitem{Haven} Khrennikov A., Haven E. Quantum mechanics and violations of
the sure-thing principle: the use of probability interference and other
concepts. \textit{Journal of Mathematical Psychology}, \textbf{53}, 378-388
(2009).

\bibitem{KHR0} A. Khrennikov, \emph{Ubiquitous quantum structure: from
psychology to finance,} Springer, Heidelberg- Berlin-New York, 2010.

\bibitem{KHR1} A. Khrennikov, \textit{Contextual approach to quantum
formalism} (Fundamental Theories of Physics). Springer, Heidelberg-
Berlin-New York, 2009.

\bibitem{AKO1} M.Asano, M.Ohya and A. Khrennikov, Quantum-Like Model for
Decision Making Process in Two Players Game, Foundations of Physics Volume
\textbf{41}, Number 3, 538-548, (2010)

\bibitem{AKO2} M.Asano, M. Ohya, Y.Tanaka, A. Khrennikov and I. Basieva, On
application of Gorini-Kossakowski-Sudarshan-Lindblad equation in
cognitive psychology, Open Systems \& Information Dynamics Volume:
\textbf{18}, Issue: 1(2011) pp. 55-69

\bibitem{BKOY} I. Basieva, A. Khrennikov, M. Ohya and I.
Yamato, Quantum-like interference effect in gene expression: glucose-lactose
destructive interference, Systems and Synthetic Biology, DOI:
10.1007/s11693-011-9081-8

\bibitem{ABKO} M.Asano, M. Ohya, Y.Tanaka, A. Khrennikov and I. Basieva,
Quantum-like Representation of Bayesian Updating, American
Institute of physics Volume 1327, pp. 57-62 : Proceedings of the
International Conference on Advances in Quantum Theory (2011)


\bibitem{PL2} A. Plotnitsky,  On the Reasonable and Unreasonable Effectiveness of 
Mathematics in Classical and Quantum Physics. Foundations of Physics.
41, 466-491 (2011).

\bibitem{GAR} C. Garola and S. Sozzo, Generalized Observables, Bell's Inequalities and Mixtures in the
ESR Model, Foundations of Physics, \textbf{41} (2011) (424-449).


\bibitem{GAR1} C. Garola and S. Sozzo, The ESR Model: A Proposal for a
Noncontextual and Local Hilbert Space Extensions of QM.
Europhysics Letters,  86,  (2009), 20009-20015.

\bibitem{Theo} Allahverdyan, A. E.,   Balian,   R., Nieuwenhuizen,  Th. M.:
The quantum measurement process in an exactly solvable model.  In:
Foundations of Probability and Physics-3, pp. 16-24. American
Institute of Physics, Ser. Conference Proceedings 750, Melville,
NY (2005).

\bibitem{DEM} De Muynck,  W. M.:  Foundations of Quantum Mechanics, an
Empiricists Approach.  Kluwer Academic Publ., Dordrecht (2002)

\bibitem{DAR} G. M. D' Ariano,  Operational axioms for quantum mechanics.
Foundations of Probability and Physics-3, pp. 79--105. American
Institute of Physics, Ser. Conference Proceedings 889, Melville,
NY (2007).

\bibitem{Fuchs} C. M. Caves, Ch. A. Fuchs, and R. Schack, 
Quantum probabilities as Bayesian probabilities. Phys. Rev. A 65, 022305 (2002).

\bibitem{Fuchs1} Ch. A. Fuchs and R. Schack, A Quantum-Bayesian Route to Quantum-State Space.
Foundations of Physics, 41,  345-356 (2011).

\bibitem{O} Ohya M.: Note on quantum proability, L.Nuovo Cimento, Vol.38,
No.11, 203-206, (1983)

\bibitem{O1} Ohya M.: On compound state and mutual information in quantum
information theory, IEEE Trans.Information Theory, 29, pp.770--777 (1983)

\bibitem{O2} Ohya M.: Complexities and their applications to
characterization of chaos, International Journal of Theoretical Physics,
Vol.37, No.1, 495-505, (1998)

\bibitem{O3} Ohya M.: Some aspects of quantum information theory and their
applications to irreversible processes, Rep.Math.Phys., Vol.27, 19-47, (1989)

\bibitem{AO2} Accardi L., Ohya M.: Compound Channels, Transition
Expectations, and Liftings, Appl. Math. Optim., 39,33-59 (1999)

\bibitem{ISO} Inoue K., Ohya M., Sato K.: Application of chaos degree to
some dynamical systems, Chaos, Soliton \& Fractals, 11, 1377-1385, (2000)

\bibitem{OV1} Ohya M., Volovich I.V.: New quantum algorithm for studying
NP-complete problems, Rep.Math.Phys., \textbf{52}, No.1,25-33 (2003)

\bibitem{OV2} Ohya.M., Volovich I.V.: Mathematical Foundations of Quantum
Information and Computation and Its Applications to Nano- and Bio-systems,
Springer, (2011)

\bibitem{IOV1} Inoue K., Ohya M. Volovich I.V.: Semiclassical properties and
chaos degree for the quantum baker's map, Journal of Mathematical Physics,
Vol.43, No.1 (2002)

\bibitem{KOT} Kossakowski A., Ohya M., Togawa Y.: How can we observe and
describe chaos?, Open System and Information Dynamics 10(3): 221-233 (2003)

\bibitem{AO1} Accardi L., Ohya M.: A Stochastic Limit Approach to the SAT
Problem,Open Systems and Information dynamics, 11,1-16, (2004)

\bibitem{O4} Ohya M.: Adaptive Dynamics and its Applications to Chaos and
NPC Problem. QP-PQ: Quantum Probability and White Noise Analysis, Quantum
Bio-Informatics 2007, 21: 181-216 (2007)

\bibitem{AIR} Accardi L.: Urne e camaleonti: Dialogo sulla realta, le leggi
del caso e la teoria quantistica.Il Saggiatore 1997 (English edition, World
Scientific 2002; japanese edition, Makino 2002, russian edition, Regular and
Chaotic dynamics (2002)

\bibitem{AOTKB} M.Asano, M. Ohya, Y.Tanaka, A. Khrennikov and I. Basieva,
Quantum-like Representation of Bayesian Updating, American Institute of
physics Volume 1327, pp. 57-62 : Proceedings of the International Conference
on Advances in Quantum Theory (2011)

\bibitem{IKA} Inada T, Kimata K, Aiba H., Mechanism responsible for
glucose-lactose diauxie in Escherichia coli challenge to the cAMP model.
Genes and Cells (1996) \textbf{1}, 293-301.

\bibitem{PRE} M.Asano, M. Ohya, Y.Tanaka, A. Khrennikov and I. Basieva,
Quantum-like Bayesian Updating, TUS preprint (2011)
\end{thebibliography}
\end{document}